%% file: main.tex
\documentclass[letterpaper,11pt]{article}
\usepackage{fullpage}
\usepackage{graphicx}
\usepackage{comment}
\usepackage{mathtools}
\usepackage[linesnumbered,ruled,vlined]{algorithm2e}
\usepackage[utf8]{inputenc}
\usepackage{tikz}
\usepackage{amsthm}
\usepackage{subcaption}
\usepackage{amsfonts}
\usepackage{hyperref}
\usepackage[framemethod=TikZ]{mdframed}
\usepackage{amsthm}
\usepackage{thmtools}
\usepackage{thm-restate}
\usepackage[normalem]{ulem}
\usepackage{float}
\usepackage{multirow}

\newtheorem{remark}{Remark}
\newtheorem{theorem}{Theorem}

\definecolor{bostonuniversityred}{rgb}{0.8, 0.0, 0.0}
\newcounter{theo}[section] 

\newcommand{\Michael}[1]{}\newcommand{\Saeed}[1]{}
\begin{document}
	
	\title{Erd\"{o}s-Szekeres Partitioning Problem}
	\author{
		Michael Mitzenmacher\thanks{Harvard University}
		\and Saeed Seddighin\thanks{Toyota Technological Institute at Chicago}
	}
	\maketitle

\begin{abstract}
	In this note, we present a substantial improvement on the computational complexity of the Erd\"{o}s-Szekeres partitioning problem and review recent works on dynamic \textsf{LIS}.
\end{abstract}

\input{src/erdos}

\bibliographystyle{abbrv}	
\bibliography{draft}

\end{document}

%% file: src/erdos.tex
\section{Erd\"{o}s-Szekeres Partitioning Problem}\label{sec:intro}
It is well-known that any sequence of size $n$ can be decomposed into $O(\sqrt{n})$ monotone subsequences. The proof follows from a simple fact: Any sequence of length $n$ contains either an increasing subsequence of length $\sqrt{n}$ or a non-increasing subsequence of length $\sqrt{n}$. Thus, one can iteratively find the maximum increasing and the maximum non-increasing subsequences of a sequence and take the larger on as one of the solution partitions. Next, by removing the partition from the original sequence and repeating this procedure with the remainder of the elements we obtain a decomposition into at most $O(\sqrt{n})$ partitions. The computational challenge is to do this in an efficient way. The above algorithm can be implemented in time $O(n^{1.5} \log n)$ if we use patience sorting in every iteration. Bar-Yehuda and Fogel~\cite{yehuda1998partitioning} improve the runtime down to $O(n^{1.5})$ by designing an algorithm that solves \textsf{LIS} in time $O(n + k^2)$ where the solution size is bounded by $k$. Since any comparison-based solution for \textsf{LIS} takes time $\Omega(n \log n)$, the gap for Erd\"{o}s-Szekeres partitioning problem has been $\Omega(\sqrt{n}/\log n)$ for quite a long time~\cite{pettie2003shortest,gronlund2014threesomes}.

We show in the following that using the recent work of Mitzenmacher and Seddighin~\cite{our-stoc-paper},  Erd\"{o}s-Szekeres partitioning problem can be solved in time $\tilde O_{\epsilon}(n^{1+\epsilon})$ for any constant $\epsilon > 0$.

\begin{theorem}\label{theorem:main}
	For any constant $\epsilon > 0$, one can in time $\tilde O_{\epsilon}(n^{1+\epsilon})$ partition any sequence of length $n$ of distinct integer numbers into $O_{\epsilon}(\sqrt{n})$ monotone (increasing or decreasing) subsequences.
\end{theorem}
\begin{proof}
	The proof follows directly from the algorithm of Mitzenmacher and Seddighin~\cite{our-stoc-paper} for dynamic \textsf{LIS}. In their setting, we start with an empty array $a$ and at every point in time we are allowed to (i) add an element, or (ii) remove an element, or (iii) substitute an element for another. The algorithm is able to update the sequence and estimate the size of the \textsf{LIS} in time $\tilde O_{\epsilon}(|a|^{\epsilon})$ where $|a|$ is the size of the array at the time the operation is performed. The approximation factor of their algorithm is constant as long as $\epsilon$ is constant. More precisely, their algorithm estimates the size of the longest increasing subsequence within a multiplicative factor of at most $(1/\epsilon)^{O(1/\epsilon)}$. Although Mitzenmacher and Seddighin~\cite{our-stoc-paper} do not explicitly state this, it implicitly follows from their algorithm that by spending additional time proportional to the reported estimation, their algorithm is able to also find an increasing subsequence with size equal to their reported length. We bring a more detailed discussion for this in Section~\ref{sec:dynamic}.
	
	Given a sequence of length $n$ with distinct numbers, we use the algorithm of Mitzenmacher and Seddighin~\cite{our-stoc-paper} to decompose it into $O_{\epsilon}(\sqrt{n})$ monotone subsequences in time $\tilde O_{\epsilon}(n^{1+\epsilon})$. To do so, we initialize two instances of their algorithm that keep an approximation to the longest increasing subsequence and the longest decreasing subsequence of the array. More precisely, in the first instance, we insert all elements of the array exactly the same way they appear in our sequence and in the second instance we insert the elements in the reverse order. Thus the dynamic algorithm for the second instance always maintains an approximation to the longest decreasing subsequence of our array.
	
	In every iteration, we estimate the size of the longest increasing and longest decreasing subsequences of the array via the algorithm of Mitzenmacher and Seddighin~\cite{our-stoc-paper}. We then choose the maximum one and ask the algorithm to give us the sequence corresponding to the solution reported. Finally, we remove the elements from both instances of the dynamic algorithm and repeat the same procedure for the remainder of the elements.
	
	The total runtime of our algorithm is $\tilde O_{\epsilon}(n^{1+\epsilon})$ since we insert $n$ elements in each of the instances and then remove $n$ elements which amount to $2n$ operations for each instance that runs in time $\tilde O_{\epsilon}(n^{1+\epsilon})$. Moreover, because at every point in time the maximum estimate we receive from each of the dynamic algorithms is at least a constant fraction of the actual longest increasing subsequence, we repeat this procedure at most $O_{\epsilon}(\sqrt{n})$ times. Therefore, we decompose the sequence into $O_{\epsilon}(\sqrt{n})$ monotone subsequences.
\end{proof}

\begin{remark}
	The constant factor hidden in the $O$ notation for the number of partitions is optimal in neither the algorithm of Theorem~\ref{theorem:main} nor previous algorithm of~\cite{yehuda1998partitioning} nor the simple greedy algorithm that runs patience sorting in every step.
\end{remark}

\section{Subsequent Work}
Since the algorithm of Mitzenmacher and Seddighin~\cite{our-soda-paper} has constant approximation factor, in order to make sure the number of partitions remains $O(\sqrt{n})$, one needs to set $\epsilon$ to constant and therefore the gap between their runtime of $\tilde O(n^{1+\epsilon})$ and the lower bound of $\Omega(n \log n)$ remains polynomial. Two independent subsequent work further tighten the gap. Kociumaka and Seddighin~\cite{saeednew} improve the gap to subpolynomial by presenting a dynamic algorithm with approximation factor $1-o(1)$ and update time $O(n^{o(1)})$. Gawrychowski and Janczewski~\cite{gawrychowski2020fully} further tighten the gap to polylogarithmic by obtaining a similar algorithm with polylogarithmic update time (with polynomial dependence on $1/\epsilon$).

\section{The Dynamic Algorithm of Mitzenmacher and Seddighin~\cite{our-stoc-paper}}\label{sec:dynamic}
In this section, we bring the high-level ideas of the dynamic algorithm of Mitzenmacher and Seddighin for \textsf{LIS} and explain why using this algorithm, we can also find the increasing subsequence corresponding to the reported solution size. Their algorithm is based on the grid packing technique explained in Section~\ref{sec:grid}. We then discuss the dynamic algorithm in Section~\ref{sec:results-approach}.

\subsection{Background: Grid Packing}\label{sec:grid}
Grid packing can be thought of as a game between us and an adversary. In this problem, we have a table of size $m \times m$. Our goal is to introduce a number of segments on the table. Each segment either covers a consecutive set of cells in a row or in a column. A segment $A$ \textit{precedes} a segment $B$ if \textbf{every} cell of $A$ is strictly higher than every cell of $B$ and also \textbf{every} cell of $A$ is strictly to the right of every cell of $B$. Two segments are \textit{non-conflicting}, if one of them precedes the other one. Otherwise, we call them \textit{conflicting}.  The segments we introduce can overlap and there is no restriction on the number of segments or the length of each segment. However, we would like to minimize the maximum number of segments that cover each cell. 

\input{figs/crossing}

After we choose the segments, an adversary puts a non-negative number on each cell of the table. The score of a subset of cells of the table would be the sum of their values and the overall score of the table is the maximum score of a path of length $2m-1$ from the bottom-left corner to the top-right corner. In such a path, we always either move up or to the right.

The score of a segment is the sum of the numbers on the cells it covers. We obtain the maximum sum of the scores of a non-conflicting set of segments.  The score of the table is an upper bound on the score of any set of non-conflicting segments. We would like to choose segments so that the ratio of the score of the table and our score is bounded by a constant, no matter how the adversary puts the numbers on the table. More precisely, we call a solution $(\alpha,\beta)$-approximate, if at most $\alpha$ segments cover each cell and it guarantees a $1/\beta$ fraction of the score of the table for us for any assignment of numbers to the table cells.

\input{figs/grid}

Mitzenmacher and Seddighin~\cite{our-stoc-paper} prove the following theorem: For any $m \times m$ table and any $0 < \kappa < 1$, there exists a grid packing solution with  guarantee $(O_{\kappa}(m^\kappa \log m),O(1/\kappa))$. That is, each cell is covered by at most $O_{\kappa}(m^\kappa \log m)$ segments and the ratio of the table's score over our score is bounded by $O(1/\kappa)$ in the worst case. 

\begin{theorem}\label{theorem:grid packing} [from~\cite{our-stoc-paper}] For any $0 < \kappa < 1$, the grid packing
	problem on an $m \times m$ table admits an $(O_{\kappa}(m^\kappa \log m),O(1/\kappa))$-approximate solution.
\end{theorem}

There is a natural connection between grid packing and \textsf{LIS}. Let us consider an array $a$ of length $n$. We assume for the sake of this example that all the numbers of the array are distinct and are in range $[1,n]$. In other words, $a$ is a permutation of numbers in $[n]$. We map the array to a set of points on the 2D plane by putting a point at $(i,a_i)$ for every position $i$ of the array.

\input{figs/lis-grid}

For a fixed $m < n$, divide the plane into an $m \times m$ grid where each row and column contains $n/m$ points. Also, we fix a longest increasing subsequence as the solution. The number on each cell of the grid would be equal to the contribution of the elements in that grid cell to the fixed longest increasing subsequence. (We emphasize that the number is {\em not} the longest increasing subsequence inside the cell, but the contribution to the fixed longest increasing subsequence only.)  It follows that the score of the grid is exactly equal to the size of the longest increasing subsequence. Let us assume that the score of each segment is available. To approximate the score of the grid (which equals the size of the \textsf{LIS}) we find the largest score we can obtain using non-conflicting segments by dynamic programming. The last observation which gives us speedup for \textsf{LIS} is the following:
instead of using the score of each segment (which we are not aware of), we use the size of the \textsf{LIS} for each segment as an approximate value for its score. \textsf{LIS} of each segment can be computed or approximated in sublinear time since each segment has a sublinear number of points. This quantity is clearly an upper bound on the score of each segment but can be used to construct a global solution for the entire array.
\input{figs/lis-grid2}

\subsection{Dynamic Algorithm for \textsf{LIS}}\label{sec:results-approach}
We refer the reader to previous work~\cite{our-stoc-paper,our-soda-paper} for discussions on how to use grid packing for approximating \textsf{LIS}. For the dynamic alogrithm, we consider the point-based representation of the problem. That is, we represent the input as points on the 2D plane where a point $(x,y)$ means the $x$'th element of the sequence has value $y$. This enables us to construct a grid where the rows and columns of the grid evenly divide the points. 
Next, we use the grid packing technique and after making the segments we construct a partial solution for each segment that keeps an approximation to the \textsf{LIS} of the points covered by that segment. Thus, every time a change is made, our algorithm has to update the solution for all segments that cover the modified point. Theorem~\ref{theorem:grid packing} implies that the number of such segments can be as small as $\tilde O(n^{\kappa})$ for any constant $\kappa > 0$. Moreover, since the number of elements covered by each segment is sublinear, the  total update time remains sublinear. 

In order to update the value of the longest increasing subsequence, Mitzenmacher and Seddighin prove that by doing a DP on the values of the partial solutions for the segments, we can obtain a constant fraction of the solution size for the entire sequence. (This essentially follows from the guarantee of Theorem~\ref{theorem:grid packing}.) They prove that the runtime of the DP depends only on the dimensions of the grid which by proper construction, results in a sublinear time algorithm. Moreover, they show that recursing on this idea leads to arbitrarily small update time $(\tilde O_{\epsilon}(n^\epsilon))$  for any constant $\epsilon > 0$. Since this approach loses a constant factor in every recursive call, their approximation factor is $\epsilon^{O(1/\epsilon)}$.

It follows from their technique that after reporting the estimated value of the solution, we can also determine the corresponding sequence in time proportional to its size. More precisely, after using DP to construct a global solution based on partial solutions of the segment, we can find out which segments contribute to such a solution and recursively recover the corresponding increasing subsequences of the relevant segments. To this end, in addition to the DP table which we use for constructing a global solution, we also store which segments contribute to such a solution. This way, the runtime required for determine the corresponding increasing subsequence is proportional to the size of the solution.

%% file: figs/crossing.tex
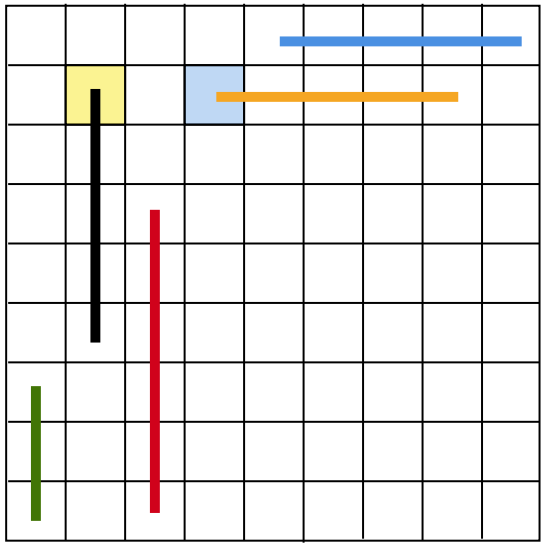
\begin{figure}[ht]

\centering

\tikzset{every picture/.style={line width=0.75pt}} %set default line width to 0.75pt        

\begin{tikzpicture}[x=0.75pt,y=0.75pt,yscale=-1,xscale=1]

%Shape: Rectangle [id:dp5254311647291401] 
\draw   (181,11) -- (450,11) -- (450,281) -- (181,281) -- cycle ;
%Straight Lines [id:da3305762449852494] 
\draw    (301,10) -- (301,281) ;

%Straight Lines [id:da3206498977891954] 
\draw    (331,10) -- (331,282) ;

%Straight Lines [id:da02838390385871148] 
\draw    (361,10) -- (361,280) ;

%Straight Lines [id:da6891763144045235] 
\draw    (211,10) -- (211,281) ;

%Straight Lines [id:da38683484497784937] 
\draw    (241,10) -- (241,281) ;

%Straight Lines [id:da7490078717616151] 
\draw    (271,10) -- (271,281) ;

%Straight Lines [id:da3531030059687228] 
\draw    (391,10) -- (391,281) ;

%Straight Lines [id:da43894336811181467] 
\draw    (421,10) -- (421,280) ;

%Straight Lines [id:da27078630886751465] 
\draw    (182,41) -- (450,41) ;

%Straight Lines [id:da005150497881268201] 
\draw    (182,71) -- (450,71) ;

%Straight Lines [id:da23147705862104506] 
\draw    (182,101) -- (450,101) ;

%Straight Lines [id:da7425549844424117] 
\draw    (182,131) -- (450,131) ;

%Straight Lines [id:da2351732964828932] 
\draw    (182,161) -- (450,161) ;

%Straight Lines [id:da7401347371733205] 
\draw    (182,191) -- (450,191) ;

%Straight Lines [id:da8968761436786192] 
\draw    (182,221) -- (450,221) ;

%Straight Lines [id:da8562431949289493] 
\draw    (182,251) -- (450,251) ;

%Straight Lines [id:da7212412833117718] 
\draw [color={rgb, 255:red, 208; green, 2; blue, 27 }  ,draw opacity=1 ][line width=3.75]    (256,114) -- (256,267) ;

%Straight Lines [id:da7675592926230685] 
\draw [color={rgb, 255:red, 74; green, 144; blue, 226 }  ,draw opacity=1 ][line width=3.75]    (319,29) -- (441,29) ;

%Straight Lines [id:da4112941771444574] 
\draw [color={rgb, 255:red, 65; green, 117; blue, 5 }  ,draw opacity=1 ][line width=3.75]    (196,203) -- (196,271) ;

%Shape: Square [id:dp4777174644258826] 
\draw  [color={rgb, 255:red, 0; green, 0; blue, 0 }  ,draw opacity=1 ][fill={rgb, 255:red, 248; green, 231; blue, 28 }  ,fill opacity=0.48 ] (211,41) -- (241,41) -- (241,71) -- (211,71) -- cycle ;
%Straight Lines [id:da09165043527230288] 
\draw [line width=3.75]    (226,53) -- (226,181) ;

%Shape: Square [id:dp8680283690742998] 
\draw  [color={rgb, 255:red, 0; green, 0; blue, 0 }  ,draw opacity=1 ][fill={rgb, 255:red, 74; green, 144; blue, 226 }  ,fill opacity=0.35 ] (271,41) -- (301,41) -- (301,71) -- (271,71) -- cycle ;
%Straight Lines [id:da0026998791284389423] 
\draw [color={rgb, 255:red, 245; green, 166; blue, 35 }  ,draw opacity=1 ][line width=3.75]    (287,57) -- (409,57) ;

\end{tikzpicture}
\caption{Segments are shown on the grid. The pair (black, orange) is conflicting since the yellow cell (covered by the black segment) is on the same row as the blue cell (covered by the orange segment). The following pairs are non-conflicting: (green, black), (green, orange), (green, blue), (red, orange), (red, blue), (black, blue).} \label{fig:crossing}
\end{figure}

%% file: figs/grid.tex
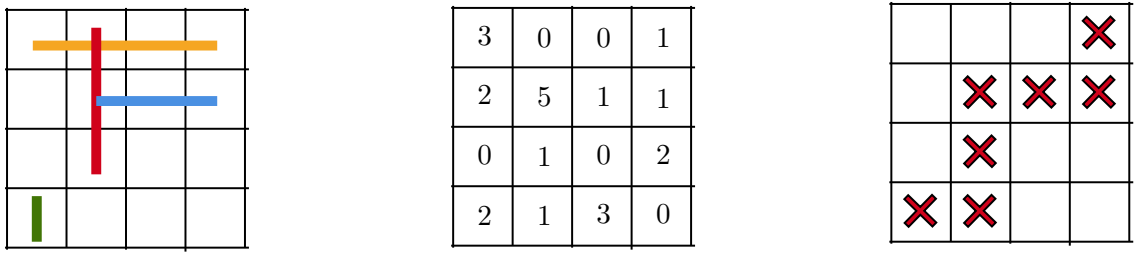
\begin{figure}[ht]

\centering

\tikzset{every picture/.style={line width=0.75pt}} %set default line width to 0.75pt        

\begin{tikzpicture}[x=0.75pt,y=0.75pt,yscale=-1,xscale=1]
%uncomment if require: \path (0,172); %set diagram left start at 0, and has height of 172

%Straight Lines [id:da6633043819092925] 
\draw    (270,30) -- (270,152) ;

%Straight Lines [id:da6069877637104277] 
\draw    (300,31) -- (300,151) ;

%Straight Lines [id:da5789537013543653] 
\draw    (330,30) -- (330,152) ;

%Straight Lines [id:da6178882799434375] 
\draw    (360,30) -- (360,153) ;

%Straight Lines [id:da33725599772480996] 
\draw    (390,30) -- (390,151) ;

%Straight Lines [id:da736940795940934] 
\draw    (269,31) -- (392,31) ;

%Straight Lines [id:da9741470082112216] 
\draw    (269,61) -- (392,61) ;

%Straight Lines [id:da2779143347704205] 
\draw    (269,91) -- (392,91) ;

%Straight Lines [id:da3530720908879399] 
\draw    (269,121) -- (392,121) ;

%Straight Lines [id:da038384016048926606] 
\draw    (269,151) -- (392,151) ;

%Straight Lines [id:da5998474774869911] 
\draw    (45,31) -- (45,153) ;

%Straight Lines [id:da5597571026286936] 
\draw    (75,32) -- (75,152) ;

%Straight Lines [id:da16479280447754507] 
\draw    (105,31) -- (105,153) ;

%Straight Lines [id:da3085712099513582] 
\draw    (135,31) -- (135,154) ;

%Straight Lines [id:da12613347387221086] 
\draw    (165,31) -- (165,152) ;

%Straight Lines [id:da6266120443422065] 
\draw    (44,32) -- (167,32) ;

%Straight Lines [id:da5072506675360497] 
\draw    (44,62) -- (167,62) ;

%Straight Lines [id:da7640683512061104] 
\draw    (44,92) -- (167,92) ;

%Straight Lines [id:da5807436513303619] 
\draw    (44,122) -- (167,122) ;

%Straight Lines [id:da2277681145302426] 
\draw    (44,152) -- (167,152) ;

%Straight Lines [id:da8599665215928698] 
\draw    (491,28) -- (491,150) ;

%Straight Lines [id:da9296198582234105] 
\draw    (521,29) -- (521,149) ;

%Straight Lines [id:da5188209661505319] 
\draw    (551,28) -- (551,150) ;

%Straight Lines [id:da7563699076629549] 
\draw    (581,28) -- (581,151) ;

%Straight Lines [id:da7100636887892002] 
\draw    (611,28) -- (611,149) ;

%Straight Lines [id:da22647026815320848] 
\draw    (490,29) -- (613,29) ;

%Straight Lines [id:da7140685594390213] 
\draw    (490,59) -- (613,59) ;

%Straight Lines [id:da7042967357998233] 
\draw    (490,89) -- (613,89) ;

%Straight Lines [id:da8328482917067017] 
\draw    (490,119) -- (613,119) ;

%Straight Lines [id:da535803664626858] 
\draw    (490,149) -- (613,149) ;

%Straight Lines [id:da2001147987658074] 
\draw [color={rgb, 255:red, 245; green, 166; blue, 35 }  ,draw opacity=1 ][line width=3.75]    (58,50) -- (151,50) ;

%Straight Lines [id:da20733782106809806] 
\draw [color={rgb, 255:red, 208; green, 2; blue, 27 }  ,draw opacity=1 ][line width=3.75]    (90,41) -- (90,115) ;

%Straight Lines [id:da3188339231354045] 
\draw [color={rgb, 255:red, 74; green, 144; blue, 226 }  ,draw opacity=1 ][line width=3.75]    (90,78) -- (151,78) ;

%Straight Lines [id:da6087723399785627] 
\draw [color={rgb, 255:red, 65; green, 117; blue, 5 }  ,draw opacity=1 ][line width=3.75]    (60,126) -- (60,149) ;

%Shape: Cross [id:dp18788217453772793] 
\draw  [color={rgb, 255:red, 0; green, 0; blue, 0 }  ,draw opacity=1 ][fill={rgb, 255:red, 208; green, 2; blue, 27 }  ,fill opacity=1 ] (497.54,127.53) -- (499.49,125.52) -- (505.58,131.45) -- (511.51,125.36) -- (513.51,127.31) -- (507.58,133.4) -- (513.67,139.33) -- (511.72,141.33) -- (505.63,135.41) -- (499.7,141.49) -- (497.7,139.54) -- (503.63,133.45) -- cycle ;
%Shape: Cross [id:dp5876950991142253] 
\draw  [color={rgb, 255:red, 0; green, 0; blue, 0 }  ,draw opacity=1 ][fill={rgb, 255:red, 208; green, 2; blue, 27 }  ,fill opacity=1 ] (527.54,127.53) -- (529.49,125.52) -- (535.58,131.45) -- (541.51,125.36) -- (543.51,127.31) -- (537.58,133.4) -- (543.67,139.33) -- (541.72,141.33) -- (535.63,135.41) -- (529.7,141.49) -- (527.7,139.54) -- (533.63,133.45) -- cycle ;
%Shape: Cross [id:dp6329349377511551] 
\draw  [color={rgb, 255:red, 0; green, 0; blue, 0 }  ,draw opacity=1 ][fill={rgb, 255:red, 208; green, 2; blue, 27 }  ,fill opacity=1 ] (527.54,97.53) -- (529.49,95.52) -- (535.58,101.45) -- (541.51,95.36) -- (543.51,97.31) -- (537.58,103.4) -- (543.67,109.33) -- (541.72,111.33) -- (535.63,105.41) -- (529.7,111.49) -- (527.7,109.54) -- (533.63,103.45) -- cycle ;
%Shape: Cross [id:dp5032748735899488] 
\draw  [color={rgb, 255:red, 0; green, 0; blue, 0 }  ,draw opacity=1 ][fill={rgb, 255:red, 208; green, 2; blue, 27 }  ,fill opacity=1 ] (527.54,67.53) -- (529.49,65.52) -- (535.58,71.45) -- (541.51,65.36) -- (543.51,67.31) -- (537.58,73.4) -- (543.67,79.33) -- (541.72,81.33) -- (535.63,75.41) -- (529.7,81.49) -- (527.7,79.54) -- (533.63,73.45) -- cycle ;
%Shape: Cross [id:dp4781774334975044] 
\draw  [color={rgb, 255:red, 0; green, 0; blue, 0 }  ,draw opacity=1 ][fill={rgb, 255:red, 208; green, 2; blue, 27 }  ,fill opacity=1 ] (557.54,67.53) -- (559.49,65.52) -- (565.58,71.45) -- (571.51,65.36) -- (573.51,67.31) -- (567.58,73.4) -- (573.67,79.33) -- (571.72,81.33) -- (565.63,75.41) -- (559.7,81.49) -- (557.7,79.54) -- (563.63,73.45) -- cycle ;
%Shape: Cross [id:dp6534331657458179] 
\draw  [color={rgb, 255:red, 0; green, 0; blue, 0 }  ,draw opacity=1 ][fill={rgb, 255:red, 208; green, 2; blue, 27 }  ,fill opacity=1 ] (587.54,67.53) -- (589.49,65.52) -- (595.58,71.45) -- (601.51,65.36) -- (603.51,67.31) -- (597.58,73.4) -- (603.67,79.33) -- (601.72,81.33) -- (595.63,75.41) -- (589.7,81.49) -- (587.7,79.54) -- (593.63,73.45) -- cycle ;
%Shape: Cross [id:dp5630155449984502] 
\draw  [color={rgb, 255:red, 0; green, 0; blue, 0 }  ,draw opacity=1 ][fill={rgb, 255:red, 208; green, 2; blue, 27 }  ,fill opacity=1 ] (587.54,37.53) -- (589.49,35.52) -- (595.58,41.45) -- (601.51,35.36) -- (603.51,37.31) -- (597.58,43.4) -- (603.67,49.33) -- (601.72,51.33) -- (595.63,45.41) -- (589.7,51.49) -- (587.7,49.54) -- (593.63,43.45) -- cycle ;

% Text Node
\draw (286,136) node   {$2$};
% Text Node
\draw (286,105) node   {$0$};
% Text Node
\draw (286,75) node   {$2$};
% Text Node
\draw (286,45) node   {$3$};
% Text Node
\draw (316,136) node   {$1$};
% Text Node
\draw (316,106) node   {$1$};
% Text Node
\draw (316,76) node   {$5$};
% Text Node
\draw (316,46) node   {$0$};
% Text Node
\draw (346,76) node   {$1$};
% Text Node
\draw (346,105) node   {$0$};
% Text Node
\draw (346,135) node   {$3$};
% Text Node
\draw (376,135) node   {$0$};
% Text Node
\draw (376,105) node   {$2$};
% Text Node
\draw (376,77) node   {$1$};
% Text Node
\draw (346,46) node   {$0$};
% Text Node
\draw (376,46) node   {$1$};

\end{tikzpicture}

\caption{After we introduce the segments (left figure), the adversary puts the numbers on the table (middle figure). In this case, the score of the table is equal to $12$ (via the path depicted on the right figure), and our score is equal to $9$ obtained from two non-conflicting segments green and blue.} \label{fig:crossing}
\end{figure}

%% file: figs/lis-grid.tex
\begin{figure}[ht]

\centering

\tikzset{every picture/.style={line width=0.75pt}} %set default line width to 0.75pt        

\begin{tikzpicture}[x=0.75pt,y=0.75pt,yscale=-1,xscale=1]
%uncomment if require: \path (0,369); %set diagram left start at 0, and has height of 369

%Straight Lines [id:da3191885316685632] 
\draw  [dash pattern={on 0.84pt off 2.51pt}]  (147,49) -- (147,320) ;

%Straight Lines [id:da7559101633943939] 
\draw  [dash pattern={on 0.84pt off 2.51pt}]  (177,49) -- (177,321) ;

%Straight Lines [id:da1291104188279324] 
\draw  [dash pattern={on 0.84pt off 2.51pt}]  (207,49) -- (207,319) ;

%Straight Lines [id:da6247914431079504] 
\draw  [dash pattern={on 0.84pt off 2.51pt}]  (57,49) -- (57,320) ;

%Straight Lines [id:da2443453671621807] 
\draw  [dash pattern={on 0.84pt off 2.51pt}]  (87,49) -- (87,320) ;

%Straight Lines [id:da4676476525067903] 
\draw  [dash pattern={on 0.84pt off 2.51pt}]  (117,49) -- (117,320) ;

%Straight Lines [id:da42635072959794496] 
\draw  [dash pattern={on 0.84pt off 2.51pt}]  (237,49) -- (237,320) ;

%Straight Lines [id:da8277079805087988] 
\draw  [dash pattern={on 0.84pt off 2.51pt}]  (267,49) -- (267,319) ;

%Straight Lines [id:da24985081014096644] 
\draw  [dash pattern={on 0.84pt off 2.51pt}]  (28,80) -- (296,80) ;

%Straight Lines [id:da6339328653912732] 
\draw  [dash pattern={on 0.84pt off 2.51pt}]  (28,110) -- (296,110) ;

%Straight Lines [id:da10731858330496191] 
\draw  [dash pattern={on 0.84pt off 2.51pt}]  (28,140) -- (296,140) ;

%Straight Lines [id:da2712941152341055] 
\draw  [dash pattern={on 0.84pt off 2.51pt}]  (28,170) -- (296,170) ;

%Straight Lines [id:da7317347402429573] 
\draw  [dash pattern={on 0.84pt off 2.51pt}]  (28,200) -- (296,200) ;

%Straight Lines [id:da8994404191047722] 
\draw  [dash pattern={on 0.84pt off 2.51pt}]  (28,230) -- (296,230) ;

%Straight Lines [id:da6044897770856661] 
\draw  [dash pattern={on 0.84pt off 2.51pt}]  (28,260) -- (296,260) ;

%Straight Lines [id:da7815256961262274] 
\draw  [dash pattern={on 0.84pt off 2.51pt}]  (28,290) -- (296,290) ;

%Straight Lines [id:da06443896163595575] 
\draw    (27,49) -- (27,320) ;

%Straight Lines [id:da8160886746264304] 
\draw  [dash pattern={on 0.84pt off 2.51pt}]  (297,49) -- (297,319) ;

%Straight Lines [id:da5343639785815872] 
\draw    (28,320) -- (296,320) ;

%Straight Lines [id:da1659083679036335] 
\draw  [dash pattern={on 0.84pt off 2.51pt}]  (28,50) -- (296,50) ;

%Shape: Ellipse [id:dp9836802741092274] 
\draw  [color={rgb, 255:red, 0; green, 0; blue, 0 }  ,draw opacity=1 ][fill={rgb, 255:red, 208; green, 2; blue, 27 }  ,fill opacity=1 ] (51,109.5) .. controls (51,106.46) and (53.46,104) .. (56.5,104) .. controls (59.54,104) and (62,106.46) .. (62,109.5) .. controls (62,112.54) and (59.54,115) .. (56.5,115) .. controls (53.46,115) and (51,112.54) .. (51,109.5) -- cycle ;
%Shape: Ellipse [id:dp8584062429307988] 
\draw  [color={rgb, 255:red, 0; green, 0; blue, 0 }  ,draw opacity=1 ][fill={rgb, 255:red, 208; green, 2; blue, 27 }  ,fill opacity=1 ] (81,259.5) .. controls (81,256.46) and (83.46,254) .. (86.5,254) .. controls (89.54,254) and (92,256.46) .. (92,259.5) .. controls (92,262.54) and (89.54,265) .. (86.5,265) .. controls (83.46,265) and (81,262.54) .. (81,259.5) -- cycle ;
%Shape: Ellipse [id:dp23277672408707772] 
\draw  [color={rgb, 255:red, 0; green, 0; blue, 0 }  ,draw opacity=1 ][fill={rgb, 255:red, 208; green, 2; blue, 27 }  ,fill opacity=1 ] (111,199.5) .. controls (111,196.46) and (113.46,194) .. (116.5,194) .. controls (119.54,194) and (122,196.46) .. (122,199.5) .. controls (122,202.54) and (119.54,205) .. (116.5,205) .. controls (113.46,205) and (111,202.54) .. (111,199.5) -- cycle ;
%Shape: Ellipse [id:dp25469902245724896] 
\draw  [color={rgb, 255:red, 0; green, 0; blue, 0 }  ,draw opacity=1 ][fill={rgb, 255:red, 208; green, 2; blue, 27 }  ,fill opacity=1 ] (141,289.5) .. controls (141,286.46) and (143.46,284) .. (146.5,284) .. controls (149.54,284) and (152,286.46) .. (152,289.5) .. controls (152,292.54) and (149.54,295) .. (146.5,295) .. controls (143.46,295) and (141,292.54) .. (141,289.5) -- cycle ;
%Shape: Ellipse [id:dp8339059316567126] 
\draw  [color={rgb, 255:red, 0; green, 0; blue, 0 }  ,draw opacity=1 ][fill={rgb, 255:red, 208; green, 2; blue, 27 }  ,fill opacity=1 ] (171,49.5) .. controls (171,46.46) and (173.46,44) .. (176.5,44) .. controls (179.54,44) and (182,46.46) .. (182,49.5) .. controls (182,52.54) and (179.54,55) .. (176.5,55) .. controls (173.46,55) and (171,52.54) .. (171,49.5) -- cycle ;
%Shape: Ellipse [id:dp8836850789925752] 
\draw  [color={rgb, 255:red, 0; green, 0; blue, 0 }  ,draw opacity=1 ][fill={rgb, 255:red, 208; green, 2; blue, 27 }  ,fill opacity=1 ] (201,139.5) .. controls (201,136.46) and (203.46,134) .. (206.5,134) .. controls (209.54,134) and (212,136.46) .. (212,139.5) .. controls (212,142.54) and (209.54,145) .. (206.5,145) .. controls (203.46,145) and (201,142.54) .. (201,139.5) -- cycle ;
%Shape: Ellipse [id:dp25676474134901084] 
\draw  [color={rgb, 255:red, 0; green, 0; blue, 0 }  ,draw opacity=1 ][fill={rgb, 255:red, 208; green, 2; blue, 27 }  ,fill opacity=1 ] (231,229.5) .. controls (231,226.46) and (233.46,224) .. (236.5,224) .. controls (239.54,224) and (242,226.46) .. (242,229.5) .. controls (242,232.54) and (239.54,235) .. (236.5,235) .. controls (233.46,235) and (231,232.54) .. (231,229.5) -- cycle ;
%Shape: Ellipse [id:dp7276624991606142] 
\draw  [color={rgb, 255:red, 0; green, 0; blue, 0 }  ,draw opacity=1 ][fill={rgb, 255:red, 208; green, 2; blue, 27 }  ,fill opacity=1 ] (261,169.5) .. controls (261,166.46) and (263.46,164) .. (266.5,164) .. controls (269.54,164) and (272,166.46) .. (272,169.5) .. controls (272,172.54) and (269.54,175) .. (266.5,175) .. controls (263.46,175) and (261,172.54) .. (261,169.5) -- cycle ;
%Shape: Ellipse [id:dp012143613453408975] 
\draw  [color={rgb, 255:red, 0; green, 0; blue, 0 }  ,draw opacity=1 ][fill={rgb, 255:red, 208; green, 2; blue, 27 }  ,fill opacity=1 ] (291,79.5) .. controls (291,76.46) and (293.46,74) .. (296.5,74) .. controls (299.54,74) and (302,76.46) .. (302,79.5) .. controls (302,82.54) and (299.54,85) .. (296.5,85) .. controls (293.46,85) and (291,82.54) .. (291,79.5) -- cycle ;
%Straight Lines [id:da5514943528211631] 
\draw [line width=3]    (417,180) -- (330,180) ;
\draw [shift={(325,180)}, rotate = 360] [fill={rgb, 255:red, 0; green, 0; blue, 0 }  ][line width=3]  [draw opacity=0] (16.97,-8.15) -- (0,0) -- (16.97,8.15) -- cycle    ;

% Text Node
\draw (-34,147) node   {$ \begin{array}{l}
	\end{array}$};
% Text Node
\draw (517,177) node   {$\langle 7,2,4,1,9,6,3,5,8\rangle$};
% Text Node
\draw (28,334) node   {$0$};
% Text Node
\draw (58,334) node   {$1$};
% Text Node
\draw (88,334) node   {$2$};
% Text Node
\draw (118,334) node   {$3$};
% Text Node
\draw (148,334) node   {$4$};
% Text Node
\draw (178,334) node   {$5$};
% Text Node
\draw (208,334) node   {$6$};
% Text Node
\draw (238,334) node   {$7$};
% Text Node
\draw (268,334) node   {$8$};
% Text Node
\draw (298,334) node   {$9$};
% Text Node
\draw (28,334) node   {$0$};
% Text Node
\draw (18,318) node   {$0$};
% Text Node
\draw (18,286) node   {$1$};
% Text Node
\draw (18,258) node   {$2$};
% Text Node
\draw (18,226) node   {$3$};
% Text Node
\draw (18,198) node   {$4$};
% Text Node
\draw (18,166) node   {$5$};
% Text Node
\draw (18,138) node   {$6$};
% Text Node
\draw (18,106) node   {$7$};
% Text Node
\draw (18,76) node   {$8$};
% Text Node
\draw (18,47) node   {$9$};

\end{tikzpicture}

\caption{An array $\langle 7, 2, 4, 1, 9, 6, 3, 5, 8\rangle$ is mapped to the 2D plane.} \label{fig:lis-grid}
\end{figure}
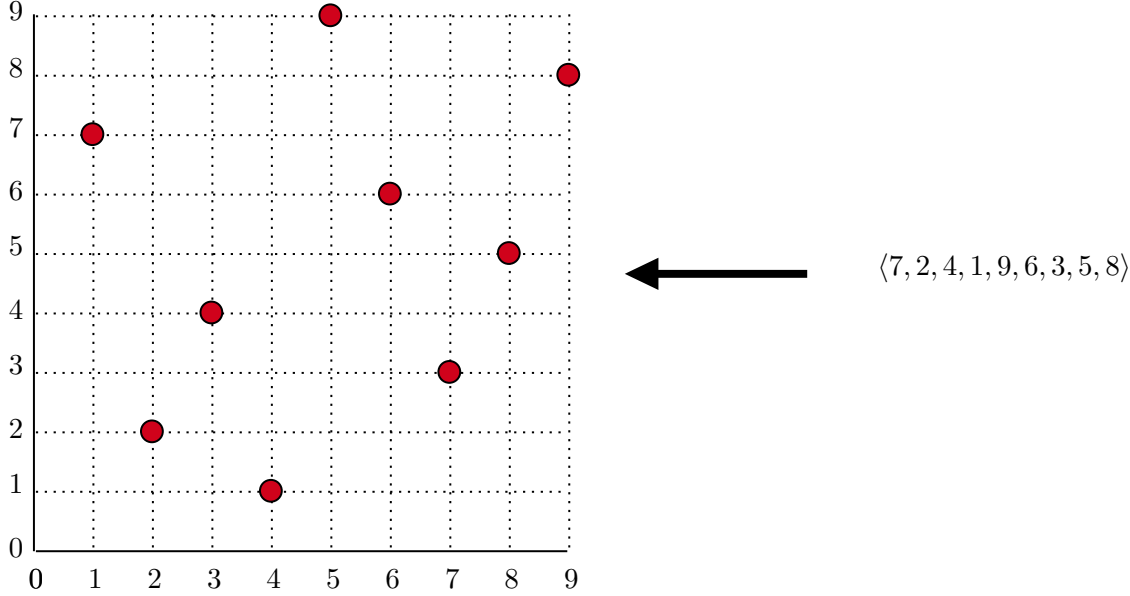

%% file: figs/lis-grid2.tex
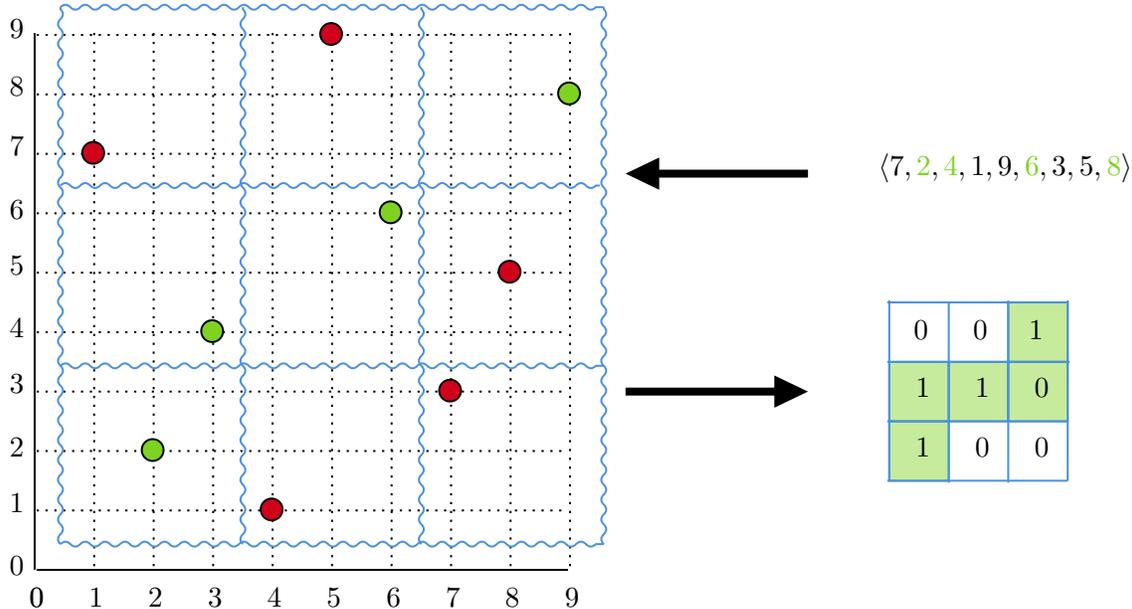
\begin{figure}[ht]

\centering

\tikzset{every picture/.style={line width=0.75pt}} %set default line width to 0.75pt        

\begin{tikzpicture}[x=0.75pt,y=0.75pt,yscale=-1,xscale=1]
%uncomment if require: \path (0,353); %set diagram left start at 0, and has height of 353

%Straight Lines [id:da8557801298612318] 
\draw  [dash pattern={on 0.84pt off 2.51pt}]  (167,35) -- (167,306) ;

%Straight Lines [id:da8067681719533357] 
\draw  [dash pattern={on 0.84pt off 2.51pt}]  (197,35) -- (197,307) ;

%Straight Lines [id:da6010871587729985] 
\draw  [dash pattern={on 0.84pt off 2.51pt}]  (227,35) -- (227,305) ;

%Straight Lines [id:da9459750229793977] 
\draw  [dash pattern={on 0.84pt off 2.51pt}]  (77,35) -- (77,306) ;

%Straight Lines [id:da39723075110564965] 
\draw  [dash pattern={on 0.84pt off 2.51pt}]  (107,35) -- (107,306) ;

%Straight Lines [id:da9496245852229939] 
\draw  [dash pattern={on 0.84pt off 2.51pt}]  (137,35) -- (137,306) ;

%Straight Lines [id:da6872631578485466] 
\draw  [dash pattern={on 0.84pt off 2.51pt}]  (257,35) -- (257,306) ;

%Straight Lines [id:da9313723680182271] 
\draw  [dash pattern={on 0.84pt off 2.51pt}]  (287,35) -- (287,305) ;

%Straight Lines [id:da2767696414454257] 
\draw  [dash pattern={on 0.84pt off 2.51pt}]  (48,66) -- (316,66) ;

%Straight Lines [id:da052820997869179154] 
\draw  [dash pattern={on 0.84pt off 2.51pt}]  (48,96) -- (316,96) ;

%Straight Lines [id:da15593366916559281] 
\draw  [dash pattern={on 0.84pt off 2.51pt}]  (48,126) -- (316,126) ;

%Straight Lines [id:da676973405819564] 
\draw  [dash pattern={on 0.84pt off 2.51pt}]  (48,156) -- (316,156) ;

%Straight Lines [id:da8486251028268585] 
\draw  [dash pattern={on 0.84pt off 2.51pt}]  (48,186) -- (316,186) ;

%Straight Lines [id:da413320555231661] 
\draw  [dash pattern={on 0.84pt off 2.51pt}]  (48,216) -- (316,216) ;

%Straight Lines [id:da7728928873185059] 
\draw  [dash pattern={on 0.84pt off 2.51pt}]  (48,246) -- (316,246) ;

%Straight Lines [id:da8123287427729207] 
\draw  [dash pattern={on 0.84pt off 2.51pt}]  (48,276) -- (316,276) ;

%Straight Lines [id:da7340260379745331] 
\draw    (47,35) -- (47,306) ;

%Straight Lines [id:da034599958348447535] 
\draw  [dash pattern={on 0.84pt off 2.51pt}]  (317,35) -- (317,305) ;

%Straight Lines [id:da7430033744449964] 
\draw    (48,306) -- (316,306) ;

%Straight Lines [id:da1040271201428109] 
\draw  [dash pattern={on 0.84pt off 2.51pt}]  (48,36) -- (316,36) ;

%Shape: Ellipse [id:dp4312639465260719] 
\draw  [color={rgb, 255:red, 0; green, 0; blue, 0 }  ,draw opacity=1 ][fill={rgb, 255:red, 208; green, 2; blue, 27 }  ,fill opacity=1 ] (71,95.5) .. controls (71,92.46) and (73.46,90) .. (76.5,90) .. controls (79.54,90) and (82,92.46) .. (82,95.5) .. controls (82,98.54) and (79.54,101) .. (76.5,101) .. controls (73.46,101) and (71,98.54) .. (71,95.5) -- cycle ;
%Shape: Ellipse [id:dp21871820242635853] 
\draw  [color={rgb, 255:red, 0; green, 0; blue, 0 }  ,draw opacity=1 ][fill={rgb, 255:red, 126; green, 211; blue, 33 }  ,fill opacity=1 ] (101,245.5) .. controls (101,242.46) and (103.46,240) .. (106.5,240) .. controls (109.54,240) and (112,242.46) .. (112,245.5) .. controls (112,248.54) and (109.54,251) .. (106.5,251) .. controls (103.46,251) and (101,248.54) .. (101,245.5) -- cycle ;
%Shape: Ellipse [id:dp4216171738035843] 
\draw  [color={rgb, 255:red, 0; green, 0; blue, 0 }  ,draw opacity=1 ][fill={rgb, 255:red, 126; green, 211; blue, 33 }  ,fill opacity=1 ] (131,185.5) .. controls (131,182.46) and (133.46,180) .. (136.5,180) .. controls (139.54,180) and (142,182.46) .. (142,185.5) .. controls (142,188.54) and (139.54,191) .. (136.5,191) .. controls (133.46,191) and (131,188.54) .. (131,185.5) -- cycle ;
%Shape: Ellipse [id:dp46388294306932654] 
\draw  [color={rgb, 255:red, 0; green, 0; blue, 0 }  ,draw opacity=1 ][fill={rgb, 255:red, 208; green, 2; blue, 27 }  ,fill opacity=1 ] (161,275.5) .. controls (161,272.46) and (163.46,270) .. (166.5,270) .. controls (169.54,270) and (172,272.46) .. (172,275.5) .. controls (172,278.54) and (169.54,281) .. (166.5,281) .. controls (163.46,281) and (161,278.54) .. (161,275.5) -- cycle ;
%Shape: Ellipse [id:dp7879797018087307] 
\draw  [color={rgb, 255:red, 0; green, 0; blue, 0 }  ,draw opacity=1 ][fill={rgb, 255:red, 208; green, 2; blue, 27 }  ,fill opacity=1 ] (191,35.5) .. controls (191,32.46) and (193.46,30) .. (196.5,30) .. controls (199.54,30) and (202,32.46) .. (202,35.5) .. controls (202,38.54) and (199.54,41) .. (196.5,41) .. controls (193.46,41) and (191,38.54) .. (191,35.5) -- cycle ;
%Shape: Ellipse [id:dp6197528996572017] 
\draw  [color={rgb, 255:red, 0; green, 0; blue, 0 }  ,draw opacity=1 ][fill={rgb, 255:red, 126; green, 211; blue, 33 }  ,fill opacity=1 ] (221,125.5) .. controls (221,122.46) and (223.46,120) .. (226.5,120) .. controls (229.54,120) and (232,122.46) .. (232,125.5) .. controls (232,128.54) and (229.54,131) .. (226.5,131) .. controls (223.46,131) and (221,128.54) .. (221,125.5) -- cycle ;
%Shape: Ellipse [id:dp3305089490263402] 
\draw  [color={rgb, 255:red, 0; green, 0; blue, 0 }  ,draw opacity=1 ][fill={rgb, 255:red, 208; green, 2; blue, 27 }  ,fill opacity=1 ] (251,215.5) .. controls (251,212.46) and (253.46,210) .. (256.5,210) .. controls (259.54,210) and (262,212.46) .. (262,215.5) .. controls (262,218.54) and (259.54,221) .. (256.5,221) .. controls (253.46,221) and (251,218.54) .. (251,215.5) -- cycle ;
%Shape: Ellipse [id:dp8842151353302732] 
\draw  [color={rgb, 255:red, 0; green, 0; blue, 0 }  ,draw opacity=1 ][fill={rgb, 255:red, 208; green, 2; blue, 27 }  ,fill opacity=1 ] (281,155.5) .. controls (281,152.46) and (283.46,150) .. (286.5,150) .. controls (289.54,150) and (292,152.46) .. (292,155.5) .. controls (292,158.54) and (289.54,161) .. (286.5,161) .. controls (283.46,161) and (281,158.54) .. (281,155.5) -- cycle ;
%Shape: Ellipse [id:dp7954844820257416] 
\draw  [color={rgb, 255:red, 0; green, 0; blue, 0 }  ,draw opacity=1 ][fill={rgb, 255:red, 126; green, 211; blue, 33 }  ,fill opacity=1 ] (311,65.5) .. controls (311,62.46) and (313.46,60) .. (316.5,60) .. controls (319.54,60) and (322,62.46) .. (322,65.5) .. controls (322,68.54) and (319.54,71) .. (316.5,71) .. controls (313.46,71) and (311,68.54) .. (311,65.5) -- cycle ;
%Straight Lines [id:da4336760374613029] 
\draw [line width=3]    (437,106) -- (350,106) ;
\draw [shift={(345,106)}, rotate = 360] [fill={rgb, 255:red, 0; green, 0; blue, 0 }  ][line width=3]  [draw opacity=0] (16.97,-8.15) -- (0,0) -- (16.97,8.15) -- cycle    ;

%Straight Lines [id:da0043566870782036915] 
\draw [color={rgb, 255:red, 74; green, 144; blue, 226 }  ,draw opacity=1 ]   (60,293) .. controls (61.67,291.33) and (63.33,291.33) .. (65,293) .. controls (66.67,294.67) and (68.33,294.67) .. (70,293) .. controls (71.67,291.33) and (73.33,291.33) .. (75,293) .. controls (76.67,294.67) and (78.33,294.67) .. (80,293) .. controls (81.67,291.33) and (83.33,291.33) .. (85,293) .. controls (86.67,294.67) and (88.33,294.67) .. (90,293) .. controls (91.67,291.33) and (93.33,291.33) .. (95,293) .. controls (96.67,294.67) and (98.33,294.67) .. (100,293) .. controls (101.67,291.33) and (103.33,291.33) .. (105,293) .. controls (106.67,294.67) and (108.33,294.67) .. (110,293) .. controls (111.67,291.33) and (113.33,291.33) .. (115,293) .. controls (116.67,294.67) and (118.33,294.67) .. (120,293) .. controls (121.67,291.33) and (123.33,291.33) .. (125,293) .. controls (126.67,294.67) and (128.33,294.67) .. (130,293) .. controls (131.67,291.33) and (133.33,291.33) .. (135,293) .. controls (136.67,294.67) and (138.33,294.67) .. (140,293) .. controls (141.67,291.33) and (143.33,291.33) .. (145,293) .. controls (146.67,294.67) and (148.33,294.67) .. (150,293) .. controls (151.67,291.33) and (153.33,291.33) .. (155,293) .. controls (156.67,294.67) and (158.33,294.67) .. (160,293) .. controls (161.67,291.33) and (163.33,291.33) .. (165,293) .. controls (166.67,294.67) and (168.33,294.67) .. (170,293) .. controls (171.67,291.33) and (173.33,291.33) .. (175,293) .. controls (176.67,294.67) and (178.33,294.67) .. (180,293) .. controls (181.67,291.33) and (183.33,291.33) .. (185,293) .. controls (186.67,294.67) and (188.33,294.67) .. (190,293) .. controls (191.67,291.33) and (193.33,291.33) .. (195,293) .. controls (196.67,294.67) and (198.33,294.67) .. (200,293) .. controls (201.67,291.33) and (203.33,291.33) .. (205,293) .. controls (206.67,294.67) and (208.33,294.67) .. (210,293) .. controls (211.67,291.33) and (213.33,291.33) .. (215,293) .. controls (216.67,294.67) and (218.33,294.67) .. (220,293) .. controls (221.67,291.33) and (223.33,291.33) .. (225,293) .. controls (226.67,294.67) and (228.33,294.67) .. (230,293) .. controls (231.67,291.33) and (233.33,291.33) .. (235,293) .. controls (236.67,294.67) and (238.33,294.67) .. (240,293) .. controls (241.67,291.33) and (243.33,291.33) .. (245,293) .. controls (246.67,294.67) and (248.33,294.67) .. (250,293) .. controls (251.67,291.33) and (253.33,291.33) .. (255,293) .. controls (256.67,294.67) and (258.33,294.67) .. (260,293) .. controls (261.67,291.33) and (263.33,291.33) .. (265,293) .. controls (266.67,294.67) and (268.33,294.67) .. (270,293) .. controls (271.67,291.33) and (273.33,291.33) .. (275,293) .. controls (276.67,294.67) and (278.33,294.67) .. (280,293) .. controls (281.67,291.33) and (283.33,291.33) .. (285,293) .. controls (286.67,294.67) and (288.33,294.67) .. (290,293) .. controls (291.67,291.33) and (293.33,291.33) .. (295,293) .. controls (296.67,294.67) and (298.33,294.67) .. (300,293) .. controls (301.67,291.33) and (303.33,291.33) .. (305,293) .. controls (306.67,294.67) and (308.33,294.67) .. (310,293) .. controls (311.67,291.33) and (313.33,291.33) .. (315,293) .. controls (316.67,294.67) and (318.33,294.67) .. (320,293) .. controls (321.67,291.33) and (323.33,291.33) .. (325,293) .. controls (326.67,294.67) and (328.33,294.67) .. (330,293) -- (334,293) -- (334,293) ;

%Straight Lines [id:da19838245981073221] 
\draw [color={rgb, 255:red, 74; green, 144; blue, 226 }  ,draw opacity=1 ]   (60,203) .. controls (61.67,201.33) and (63.33,201.33) .. (65,203) .. controls (66.67,204.67) and (68.33,204.67) .. (70,203) .. controls (71.67,201.33) and (73.33,201.33) .. (75,203) .. controls (76.67,204.67) and (78.33,204.67) .. (80,203) .. controls (81.67,201.33) and (83.33,201.33) .. (85,203) .. controls (86.67,204.67) and (88.33,204.67) .. (90,203) .. controls (91.67,201.33) and (93.33,201.33) .. (95,203) .. controls (96.67,204.67) and (98.33,204.67) .. (100,203) .. controls (101.67,201.33) and (103.33,201.33) .. (105,203) .. controls (106.67,204.67) and (108.33,204.67) .. (110,203) .. controls (111.67,201.33) and (113.33,201.33) .. (115,203) .. controls (116.67,204.67) and (118.33,204.67) .. (120,203) .. controls (121.67,201.33) and (123.33,201.33) .. (125,203) .. controls (126.67,204.67) and (128.33,204.67) .. (130,203) .. controls (131.67,201.33) and (133.33,201.33) .. (135,203) .. controls (136.67,204.67) and (138.33,204.67) .. (140,203) .. controls (141.67,201.33) and (143.33,201.33) .. (145,203) .. controls (146.67,204.67) and (148.33,204.67) .. (150,203) .. controls (151.67,201.33) and (153.33,201.33) .. (155,203) .. controls (156.67,204.67) and (158.33,204.67) .. (160,203) .. controls (161.67,201.33) and (163.33,201.33) .. (165,203) .. controls (166.67,204.67) and (168.33,204.67) .. (170,203) .. controls (171.67,201.33) and (173.33,201.33) .. (175,203) .. controls (176.67,204.67) and (178.33,204.67) .. (180,203) .. controls (181.67,201.33) and (183.33,201.33) .. (185,203) .. controls (186.67,204.67) and (188.33,204.67) .. (190,203) .. controls (191.67,201.33) and (193.33,201.33) .. (195,203) .. controls (196.67,204.67) and (198.33,204.67) .. (200,203) .. controls (201.67,201.33) and (203.33,201.33) .. (205,203) .. controls (206.67,204.67) and (208.33,204.67) .. (210,203) .. controls (211.67,201.33) and (213.33,201.33) .. (215,203) .. controls (216.67,204.67) and (218.33,204.67) .. (220,203) .. controls (221.67,201.33) and (223.33,201.33) .. (225,203) .. controls (226.67,204.67) and (228.33,204.67) .. (230,203) .. controls (231.67,201.33) and (233.33,201.33) .. (235,203) .. controls (236.67,204.67) and (238.33,204.67) .. (240,203) .. controls (241.67,201.33) and (243.33,201.33) .. (245,203) .. controls (246.67,204.67) and (248.33,204.67) .. (250,203) .. controls (251.67,201.33) and (253.33,201.33) .. (255,203) .. controls (256.67,204.67) and (258.33,204.67) .. (260,203) .. controls (261.67,201.33) and (263.33,201.33) .. (265,203) .. controls (266.67,204.67) and (268.33,204.67) .. (270,203) .. controls (271.67,201.33) and (273.33,201.33) .. (275,203) .. controls (276.67,204.67) and (278.33,204.67) .. (280,203) .. controls (281.67,201.33) and (283.33,201.33) .. (285,203) .. controls (286.67,204.67) and (288.33,204.67) .. (290,203) .. controls (291.67,201.33) and (293.33,201.33) .. (295,203) .. controls (296.67,204.67) and (298.33,204.67) .. (300,203) .. controls (301.67,201.33) and (303.33,201.33) .. (305,203) .. controls (306.67,204.67) and (308.33,204.67) .. (310,203) .. controls (311.67,201.33) and (313.33,201.33) .. (315,203) .. controls (316.67,204.67) and (318.33,204.67) .. (320,203) .. controls (321.67,201.33) and (323.33,201.33) .. (325,203) .. controls (326.67,204.67) and (328.33,204.67) .. (330,203) -- (334,203) -- (334,203) ;

%Straight Lines [id:da9390034698517418] 
\draw [color={rgb, 255:red, 74; green, 144; blue, 226 }  ,draw opacity=1 ]   (60,112) .. controls (61.67,110.33) and (63.33,110.33) .. (65,112) .. controls (66.67,113.67) and (68.33,113.67) .. (70,112) .. controls (71.67,110.33) and (73.33,110.33) .. (75,112) .. controls (76.67,113.67) and (78.33,113.67) .. (80,112) .. controls (81.67,110.33) and (83.33,110.33) .. (85,112) .. controls (86.67,113.67) and (88.33,113.67) .. (90,112) .. controls (91.67,110.33) and (93.33,110.33) .. (95,112) .. controls (96.67,113.67) and (98.33,113.67) .. (100,112) .. controls (101.67,110.33) and (103.33,110.33) .. (105,112) .. controls (106.67,113.67) and (108.33,113.67) .. (110,112) .. controls (111.67,110.33) and (113.33,110.33) .. (115,112) .. controls (116.67,113.67) and (118.33,113.67) .. (120,112) .. controls (121.67,110.33) and (123.33,110.33) .. (125,112) .. controls (126.67,113.67) and (128.33,113.67) .. (130,112) .. controls (131.67,110.33) and (133.33,110.33) .. (135,112) .. controls (136.67,113.67) and (138.33,113.67) .. (140,112) .. controls (141.67,110.33) and (143.33,110.33) .. (145,112) .. controls (146.67,113.67) and (148.33,113.67) .. (150,112) .. controls (151.67,110.33) and (153.33,110.33) .. (155,112) .. controls (156.67,113.67) and (158.33,113.67) .. (160,112) .. controls (161.67,110.33) and (163.33,110.33) .. (165,112) .. controls (166.67,113.67) and (168.33,113.67) .. (170,112) .. controls (171.67,110.33) and (173.33,110.33) .. (175,112) .. controls (176.67,113.67) and (178.33,113.67) .. (180,112) .. controls (181.67,110.33) and (183.33,110.33) .. (185,112) .. controls (186.67,113.67) and (188.33,113.67) .. (190,112) .. controls (191.67,110.33) and (193.33,110.33) .. (195,112) .. controls (196.67,113.67) and (198.33,113.67) .. (200,112) .. controls (201.67,110.33) and (203.33,110.33) .. (205,112) .. controls (206.67,113.67) and (208.33,113.67) .. (210,112) .. controls (211.67,110.33) and (213.33,110.33) .. (215,112) .. controls (216.67,113.67) and (218.33,113.67) .. (220,112) .. controls (221.67,110.33) and (223.33,110.33) .. (225,112) .. controls (226.67,113.67) and (228.33,113.67) .. (230,112) .. controls (231.67,110.33) and (233.33,110.33) .. (235,112) .. controls (236.67,113.67) and (238.33,113.67) .. (240,112) .. controls (241.67,110.33) and (243.33,110.33) .. (245,112) .. controls (246.67,113.67) and (248.33,113.67) .. (250,112) .. controls (251.67,110.33) and (253.33,110.33) .. (255,112) .. controls (256.67,113.67) and (258.33,113.67) .. (260,112) .. controls (261.67,110.33) and (263.33,110.33) .. (265,112) .. controls (266.67,113.67) and (268.33,113.67) .. (270,112) .. controls (271.67,110.33) and (273.33,110.33) .. (275,112) .. controls (276.67,113.67) and (278.33,113.67) .. (280,112) .. controls (281.67,110.33) and (283.33,110.33) .. (285,112) .. controls (286.67,113.67) and (288.33,113.67) .. (290,112) .. controls (291.67,110.33) and (293.33,110.33) .. (295,112) .. controls (296.67,113.67) and (298.33,113.67) .. (300,112) .. controls (301.67,110.33) and (303.33,110.33) .. (305,112) .. controls (306.67,113.67) and (308.33,113.67) .. (310,112) .. controls (311.67,110.33) and (313.33,110.33) .. (315,112) .. controls (316.67,113.67) and (318.33,113.67) .. (320,112) .. controls (321.67,110.33) and (323.33,110.33) .. (325,112) .. controls (326.67,113.67) and (328.33,113.67) .. (330,112) -- (332,112) -- (332,112) ;

%Straight Lines [id:da776057826693265] 
\draw [color={rgb, 255:red, 74; green, 144; blue, 226 }  ,draw opacity=1 ]   (60,22) .. controls (61.67,20.33) and (63.33,20.33) .. (65,22) .. controls (66.67,23.67) and (68.33,23.67) .. (70,22) .. controls (71.67,20.33) and (73.33,20.33) .. (75,22) .. controls (76.67,23.67) and (78.33,23.67) .. (80,22) .. controls (81.67,20.33) and (83.33,20.33) .. (85,22) .. controls (86.67,23.67) and (88.33,23.67) .. (90,22) .. controls (91.67,20.33) and (93.33,20.33) .. (95,22) .. controls (96.67,23.67) and (98.33,23.67) .. (100,22) .. controls (101.67,20.33) and (103.33,20.33) .. (105,22) .. controls (106.67,23.67) and (108.33,23.67) .. (110,22) .. controls (111.67,20.33) and (113.33,20.33) .. (115,22) .. controls (116.67,23.67) and (118.33,23.67) .. (120,22) .. controls (121.67,20.33) and (123.33,20.33) .. (125,22) .. controls (126.67,23.67) and (128.33,23.67) .. (130,22) .. controls (131.67,20.33) and (133.33,20.33) .. (135,22) .. controls (136.67,23.67) and (138.33,23.67) .. (140,22) .. controls (141.67,20.33) and (143.33,20.33) .. (145,22) .. controls (146.67,23.67) and (148.33,23.67) .. (150,22) .. controls (151.67,20.33) and (153.33,20.33) .. (155,22) .. controls (156.67,23.67) and (158.33,23.67) .. (160,22) .. controls (161.67,20.33) and (163.33,20.33) .. (165,22) .. controls (166.67,23.67) and (168.33,23.67) .. (170,22) .. controls (171.67,20.33) and (173.33,20.33) .. (175,22) .. controls (176.67,23.67) and (178.33,23.67) .. (180,22) .. controls (181.67,20.33) and (183.33,20.33) .. (185,22) .. controls (186.67,23.67) and (188.33,23.67) .. (190,22) .. controls (191.67,20.33) and (193.33,20.33) .. (195,22) .. controls (196.67,23.67) and (198.33,23.67) .. (200,22) .. controls (201.67,20.33) and (203.33,20.33) .. (205,22) .. controls (206.67,23.67) and (208.33,23.67) .. (210,22) .. controls (211.67,20.33) and (213.33,20.33) .. (215,22) .. controls (216.67,23.67) and (218.33,23.67) .. (220,22) .. controls (221.67,20.33) and (223.33,20.33) .. (225,22) .. controls (226.67,23.67) and (228.33,23.67) .. (230,22) .. controls (231.67,20.33) and (233.33,20.33) .. (235,22) .. controls (236.67,23.67) and (238.33,23.67) .. (240,22) .. controls (241.67,20.33) and (243.33,20.33) .. (245,22) .. controls (246.67,23.67) and (248.33,23.67) .. (250,22) .. controls (251.67,20.33) and (253.33,20.33) .. (255,22) .. controls (256.67,23.67) and (258.33,23.67) .. (260,22) .. controls (261.67,20.33) and (263.33,20.33) .. (265,22) .. controls (266.67,23.67) and (268.33,23.67) .. (270,22) .. controls (271.67,20.33) and (273.33,20.33) .. (275,22) .. controls (276.67,23.67) and (278.33,23.67) .. (280,22) .. controls (281.67,20.33) and (283.33,20.33) .. (285,22) .. controls (286.67,23.67) and (288.33,23.67) .. (290,22) .. controls (291.67,20.33) and (293.33,20.33) .. (295,22) .. controls (296.67,23.67) and (298.33,23.67) .. (300,22) .. controls (301.67,20.33) and (303.33,20.33) .. (305,22) .. controls (306.67,23.67) and (308.33,23.67) .. (310,22) .. controls (311.67,20.33) and (313.33,20.33) .. (315,22) .. controls (316.67,23.67) and (318.33,23.67) .. (320,22) .. controls (321.67,20.33) and (323.33,20.33) .. (325,22) .. controls (326.67,23.67) and (328.33,23.67) .. (330,22) -- (334,22) -- (334,22) ;

%Straight Lines [id:da5054599020816195] 
\draw [color={rgb, 255:red, 74; green, 144; blue, 226 }  ,draw opacity=1 ]   (60,293) .. controls (58.33,291.33) and (58.33,289.67) .. (60,288) .. controls (61.67,286.33) and (61.67,284.67) .. (60,283) .. controls (58.33,281.33) and (58.33,279.67) .. (60,278) .. controls (61.67,276.33) and (61.67,274.67) .. (60,273) .. controls (58.33,271.33) and (58.33,269.67) .. (60,268) .. controls (61.67,266.33) and (61.67,264.67) .. (60,263) .. controls (58.33,261.33) and (58.33,259.67) .. (60,258) .. controls (61.67,256.33) and (61.67,254.67) .. (60,253) .. controls (58.33,251.33) and (58.33,249.67) .. (60,248) .. controls (61.67,246.33) and (61.67,244.67) .. (60,243) .. controls (58.33,241.33) and (58.33,239.67) .. (60,238) .. controls (61.67,236.33) and (61.67,234.67) .. (60,233) .. controls (58.33,231.33) and (58.33,229.67) .. (60,228) .. controls (61.67,226.33) and (61.67,224.67) .. (60,223) .. controls (58.33,221.33) and (58.33,219.67) .. (60,218) .. controls (61.67,216.33) and (61.67,214.67) .. (60,213) .. controls (58.33,211.33) and (58.33,209.67) .. (60,208) .. controls (61.67,206.33) and (61.67,204.67) .. (60,203) .. controls (58.33,201.33) and (58.33,199.67) .. (60,198) .. controls (61.67,196.33) and (61.67,194.67) .. (60,193) .. controls (58.33,191.33) and (58.33,189.67) .. (60,188) .. controls (61.67,186.33) and (61.67,184.67) .. (60,183) .. controls (58.33,181.33) and (58.33,179.67) .. (60,178) .. controls (61.67,176.33) and (61.67,174.67) .. (60,173) .. controls (58.33,171.33) and (58.33,169.67) .. (60,168) .. controls (61.67,166.33) and (61.67,164.67) .. (60,163) .. controls (58.33,161.33) and (58.33,159.67) .. (60,158) .. controls (61.67,156.33) and (61.67,154.67) .. (60,153) .. controls (58.33,151.33) and (58.33,149.67) .. (60,148) .. controls (61.67,146.33) and (61.67,144.67) .. (60,143) .. controls (58.33,141.33) and (58.33,139.67) .. (60,138) .. controls (61.67,136.33) and (61.67,134.67) .. (60,133) .. controls (58.33,131.33) and (58.33,129.67) .. (60,128) .. controls (61.67,126.33) and (61.67,124.67) .. (60,123) .. controls (58.33,121.33) and (58.33,119.67) .. (60,118) .. controls (61.67,116.33) and (61.67,114.67) .. (60,113) .. controls (58.33,111.33) and (58.33,109.67) .. (60,108) .. controls (61.67,106.33) and (61.67,104.67) .. (60,103) .. controls (58.33,101.33) and (58.33,99.67) .. (60,98) .. controls (61.67,96.33) and (61.67,94.67) .. (60,93) .. controls (58.33,91.33) and (58.33,89.67) .. (60,88) .. controls (61.67,86.33) and (61.67,84.67) .. (60,83) .. controls (58.33,81.33) and (58.33,79.67) .. (60,78) .. controls (61.67,76.33) and (61.67,74.67) .. (60,73) .. controls (58.33,71.33) and (58.33,69.67) .. (60,68) .. controls (61.67,66.33) and (61.67,64.67) .. (60,63) .. controls (58.33,61.33) and (58.33,59.67) .. (60,58) .. controls (61.67,56.33) and (61.67,54.67) .. (60,53) .. controls (58.33,51.33) and (58.33,49.67) .. (60,48) .. controls (61.67,46.33) and (61.67,44.67) .. (60,43) .. controls (58.33,41.33) and (58.33,39.67) .. (60,38) .. controls (61.67,36.33) and (61.67,34.67) .. (60,33) .. controls (58.33,31.33) and (58.33,29.67) .. (60,28) .. controls (61.67,26.33) and (61.67,24.67) .. (60,23) -- (60,22) -- (60,22) ;

%Straight Lines [id:da6070775423554995] 
\draw [color={rgb, 255:red, 74; green, 144; blue, 226 }  ,draw opacity=1 ]   (152,293) .. controls (150.33,291.33) and (150.33,289.67) .. (152,288) .. controls (153.67,286.33) and (153.67,284.67) .. (152,283) .. controls (150.33,281.33) and (150.33,279.67) .. (152,278) .. controls (153.67,276.33) and (153.67,274.67) .. (152,273) .. controls (150.33,271.33) and (150.33,269.67) .. (152,268) .. controls (153.67,266.33) and (153.67,264.67) .. (152,263) .. controls (150.33,261.33) and (150.33,259.67) .. (152,258) .. controls (153.67,256.33) and (153.67,254.67) .. (152,253) .. controls (150.33,251.33) and (150.33,249.67) .. (152,248) .. controls (153.67,246.33) and (153.67,244.67) .. (152,243) .. controls (150.33,241.33) and (150.33,239.67) .. (152,238) .. controls (153.67,236.33) and (153.67,234.67) .. (152,233) .. controls (150.33,231.33) and (150.33,229.67) .. (152,228) .. controls (153.67,226.33) and (153.67,224.67) .. (152,223) .. controls (150.33,221.33) and (150.33,219.67) .. (152,218) .. controls (153.67,216.33) and (153.67,214.67) .. (152,213) .. controls (150.33,211.33) and (150.33,209.67) .. (152,208) .. controls (153.67,206.33) and (153.67,204.67) .. (152,203) .. controls (150.33,201.33) and (150.33,199.67) .. (152,198) .. controls (153.67,196.33) and (153.67,194.67) .. (152,193) .. controls (150.33,191.33) and (150.33,189.67) .. (152,188) .. controls (153.67,186.33) and (153.67,184.67) .. (152,183) .. controls (150.33,181.33) and (150.33,179.67) .. (152,178) .. controls (153.67,176.33) and (153.67,174.67) .. (152,173) .. controls (150.33,171.33) and (150.33,169.67) .. (152,168) .. controls (153.67,166.33) and (153.67,164.67) .. (152,163) .. controls (150.33,161.33) and (150.33,159.67) .. (152,158) .. controls (153.67,156.33) and (153.67,154.67) .. (152,153) .. controls (150.33,151.33) and (150.33,149.67) .. (152,148) .. controls (153.67,146.33) and (153.67,144.67) .. (152,143) .. controls (150.33,141.33) and (150.33,139.67) .. (152,138) .. controls (153.67,136.33) and (153.67,134.67) .. (152,133) .. controls (150.33,131.33) and (150.33,129.67) .. (152,128) .. controls (153.67,126.33) and (153.67,124.67) .. (152,123) .. controls (150.33,121.33) and (150.33,119.67) .. (152,118) .. controls (153.67,116.33) and (153.67,114.67) .. (152,113) .. controls (150.33,111.33) and (150.33,109.67) .. (152,108) .. controls (153.67,106.33) and (153.67,104.67) .. (152,103) .. controls (150.33,101.33) and (150.33,99.67) .. (152,98) .. controls (153.67,96.33) and (153.67,94.67) .. (152,93) .. controls (150.33,91.33) and (150.33,89.67) .. (152,88) .. controls (153.67,86.33) and (153.67,84.67) .. (152,83) .. controls (150.33,81.33) and (150.33,79.67) .. (152,78) .. controls (153.67,76.33) and (153.67,74.67) .. (152,73) .. controls (150.33,71.33) and (150.33,69.67) .. (152,68) .. controls (153.67,66.33) and (153.67,64.67) .. (152,63) .. controls (150.33,61.33) and (150.33,59.67) .. (152,58) .. controls (153.67,56.33) and (153.67,54.67) .. (152,53) .. controls (150.33,51.33) and (150.33,49.67) .. (152,48) .. controls (153.67,46.33) and (153.67,44.67) .. (152,43) .. controls (150.33,41.33) and (150.33,39.67) .. (152,38) .. controls (153.67,36.33) and (153.67,34.67) .. (152,33) .. controls (150.33,31.33) and (150.33,29.67) .. (152,28) .. controls (153.67,26.33) and (153.67,24.67) .. (152,23) -- (152,22) -- (152,22) ;

%Straight Lines [id:da3815754513530163] 
\draw [color={rgb, 255:red, 74; green, 144; blue, 226 }  ,draw opacity=1 ]   (242,293) .. controls (240.33,291.33) and (240.33,289.67) .. (242,288) .. controls (243.67,286.33) and (243.67,284.67) .. (242,283) .. controls (240.33,281.33) and (240.33,279.67) .. (242,278) .. controls (243.67,276.33) and (243.67,274.67) .. (242,273) .. controls (240.33,271.33) and (240.33,269.67) .. (242,268) .. controls (243.67,266.33) and (243.67,264.67) .. (242,263) .. controls (240.33,261.33) and (240.33,259.67) .. (242,258) .. controls (243.67,256.33) and (243.67,254.67) .. (242,253) .. controls (240.33,251.33) and (240.33,249.67) .. (242,248) .. controls (243.67,246.33) and (243.67,244.67) .. (242,243) .. controls (240.33,241.33) and (240.33,239.67) .. (242,238) .. controls (243.67,236.33) and (243.67,234.67) .. (242,233) .. controls (240.33,231.33) and (240.33,229.67) .. (242,228) .. controls (243.67,226.33) and (243.67,224.67) .. (242,223) .. controls (240.33,221.33) and (240.33,219.67) .. (242,218) .. controls (243.67,216.33) and (243.67,214.67) .. (242,213) .. controls (240.33,211.33) and (240.33,209.67) .. (242,208) .. controls (243.67,206.33) and (243.67,204.67) .. (242,203) .. controls (240.33,201.33) and (240.33,199.67) .. (242,198) .. controls (243.67,196.33) and (243.67,194.67) .. (242,193) .. controls (240.33,191.33) and (240.33,189.67) .. (242,188) .. controls (243.67,186.33) and (243.67,184.67) .. (242,183) .. controls (240.33,181.33) and (240.33,179.67) .. (242,178) .. controls (243.67,176.33) and (243.67,174.67) .. (242,173) .. controls (240.33,171.33) and (240.33,169.67) .. (242,168) .. controls (243.67,166.33) and (243.67,164.67) .. (242,163) .. controls (240.33,161.33) and (240.33,159.67) .. (242,158) .. controls (243.67,156.33) and (243.67,154.67) .. (242,153) .. controls (240.33,151.33) and (240.33,149.67) .. (242,148) .. controls (243.67,146.33) and (243.67,144.67) .. (242,143) .. controls (240.33,141.33) and (240.33,139.67) .. (242,138) .. controls (243.67,136.33) and (243.67,134.67) .. (242,133) .. controls (240.33,131.33) and (240.33,129.67) .. (242,128) .. controls (243.67,126.33) and (243.67,124.67) .. (242,123) .. controls (240.33,121.33) and (240.33,119.67) .. (242,118) .. controls (243.67,116.33) and (243.67,114.67) .. (242,113) .. controls (240.33,111.33) and (240.33,109.67) .. (242,108) .. controls (243.67,106.33) and (243.67,104.67) .. (242,103) .. controls (240.33,101.33) and (240.33,99.67) .. (242,98) .. controls (243.67,96.33) and (243.67,94.67) .. (242,93) .. controls (240.33,91.33) and (240.33,89.67) .. (242,88) .. controls (243.67,86.33) and (243.67,84.67) .. (242,83) .. controls (240.33,81.33) and (240.33,79.67) .. (242,78) .. controls (243.67,76.33) and (243.67,74.67) .. (242,73) .. controls (240.33,71.33) and (240.33,69.67) .. (242,68) .. controls (243.67,66.33) and (243.67,64.67) .. (242,63) .. controls (240.33,61.33) and (240.33,59.67) .. (242,58) .. controls (243.67,56.33) and (243.67,54.67) .. (242,53) .. controls (240.33,51.33) and (240.33,49.67) .. (242,48) .. controls (243.67,46.33) and (243.67,44.67) .. (242,43) .. controls (240.33,41.33) and (240.33,39.67) .. (242,38) .. controls (243.67,36.33) and (243.67,34.67) .. (242,33) .. controls (240.33,31.33) and (240.33,29.67) .. (242,28) .. controls (243.67,26.33) and (243.67,24.67) .. (242,23) -- (242,22) -- (242,22) ;

%Straight Lines [id:da22288970820132348] 
\draw [color={rgb, 255:red, 74; green, 144; blue, 226 }  ,draw opacity=1 ]   (334,293) .. controls (332.33,291.33) and (332.33,289.67) .. (334,288) .. controls (335.67,286.33) and (335.67,284.67) .. (334,283) .. controls (332.33,281.33) and (332.33,279.67) .. (334,278) .. controls (335.67,276.33) and (335.67,274.67) .. (334,273) .. controls (332.33,271.33) and (332.33,269.67) .. (334,268) .. controls (335.67,266.33) and (335.67,264.67) .. (334,263) .. controls (332.33,261.33) and (332.33,259.67) .. (334,258) .. controls (335.67,256.33) and (335.67,254.67) .. (334,253) .. controls (332.33,251.33) and (332.33,249.67) .. (334,248) .. controls (335.67,246.33) and (335.67,244.67) .. (334,243) .. controls (332.33,241.33) and (332.33,239.67) .. (334,238) .. controls (335.67,236.33) and (335.67,234.67) .. (334,233) .. controls (332.33,231.33) and (332.33,229.67) .. (334,228) .. controls (335.67,226.33) and (335.67,224.67) .. (334,223) .. controls (332.33,221.33) and (332.33,219.67) .. (334,218) .. controls (335.67,216.33) and (335.67,214.67) .. (334,213) .. controls (332.33,211.33) and (332.33,209.67) .. (334,208) .. controls (335.67,206.33) and (335.67,204.67) .. (334,203) .. controls (332.33,201.33) and (332.33,199.67) .. (334,198) .. controls (335.67,196.33) and (335.67,194.67) .. (334,193) .. controls (332.33,191.33) and (332.33,189.67) .. (334,188) .. controls (335.67,186.33) and (335.67,184.67) .. (334,183) .. controls (332.33,181.33) and (332.33,179.67) .. (334,178) .. controls (335.67,176.33) and (335.67,174.67) .. (334,173) .. controls (332.33,171.33) and (332.33,169.67) .. (334,168) .. controls (335.67,166.33) and (335.67,164.67) .. (334,163) .. controls (332.33,161.33) and (332.33,159.67) .. (334,158) .. controls (335.67,156.33) and (335.67,154.67) .. (334,153) .. controls (332.33,151.33) and (332.33,149.67) .. (334,148) .. controls (335.67,146.33) and (335.67,144.67) .. (334,143) .. controls (332.33,141.33) and (332.33,139.67) .. (334,138) .. controls (335.67,136.33) and (335.67,134.67) .. (334,133) .. controls (332.33,131.33) and (332.33,129.67) .. (334,128) .. controls (335.67,126.33) and (335.67,124.67) .. (334,123) .. controls (332.33,121.33) and (332.33,119.67) .. (334,118) .. controls (335.67,116.33) and (335.67,114.67) .. (334,113) .. controls (332.33,111.33) and (332.33,109.67) .. (334,108) .. controls (335.67,106.33) and (335.67,104.67) .. (334,103) .. controls (332.33,101.33) and (332.33,99.67) .. (334,98) .. controls (335.67,96.33) and (335.67,94.67) .. (334,93) .. controls (332.33,91.33) and (332.33,89.67) .. (334,88) .. controls (335.67,86.33) and (335.67,84.67) .. (334,83) .. controls (332.33,81.33) and (332.33,79.67) .. (334,78) .. controls (335.67,76.33) and (335.67,74.67) .. (334,73) .. controls (332.33,71.33) and (332.33,69.67) .. (334,68) .. controls (335.67,66.33) and (335.67,64.67) .. (334,63) .. controls (332.33,61.33) and (332.33,59.67) .. (334,58) .. controls (335.67,56.33) and (335.67,54.67) .. (334,53) .. controls (332.33,51.33) and (332.33,49.67) .. (334,48) .. controls (335.67,46.33) and (335.67,44.67) .. (334,43) .. controls (332.33,41.33) and (332.33,39.67) .. (334,38) .. controls (335.67,36.33) and (335.67,34.67) .. (334,33) .. controls (332.33,31.33) and (332.33,29.67) .. (334,28) .. controls (335.67,26.33) and (335.67,24.67) .. (334,23) -- (334,22) -- (334,22) ;

%Straight Lines [id:da7606176214865035] 
\draw [line width=3]    (432,216) -- (345,216) ;

\draw [shift={(437,216)}, rotate = 180] [fill={rgb, 255:red, 0; green, 0; blue, 0 }  ][line width=3]  [draw opacity=0] (16.97,-8.15) -- (0,0) -- (16.97,8.15) -- cycle    ;
%Straight Lines [id:da3018322912330984] 
\draw [color={rgb, 255:red, 74; green, 144; blue, 226 }  ,draw opacity=1 ]   (478,170) -- (478,261) ;

%Straight Lines [id:da015607019863883131] 
\draw [color={rgb, 255:red, 74; green, 144; blue, 226 }  ,draw opacity=1 ]   (508,171) -- (508,262) ;

%Straight Lines [id:da35809350995635225] 
\draw [color={rgb, 255:red, 74; green, 144; blue, 226 }  ,draw opacity=1 ]   (538,170) -- (538,262) ;

%Straight Lines [id:da4663748961004479] 
\draw [color={rgb, 255:red, 74; green, 144; blue, 226 }  ,draw opacity=1 ]   (568,170) -- (568,261) ;

%Straight Lines [id:da24669242453466422] 
\draw [color={rgb, 255:red, 74; green, 144; blue, 226 }  ,draw opacity=1 ]   (477,171) -- (568,171) ;

%Straight Lines [id:da09726362129042587] 
\draw [color={rgb, 255:red, 74; green, 144; blue, 226 }  ,draw opacity=1 ]   (477,201) -- (569,201) ;

%Straight Lines [id:da5730679880179985] 
\draw [color={rgb, 255:red, 74; green, 144; blue, 226 }  ,draw opacity=1 ]   (477,231) -- (568,231) ;

%Straight Lines [id:da49079358203415246] 
\draw [color={rgb, 255:red, 74; green, 144; blue, 226 }  ,draw opacity=1 ]   (477,261) -- (568,261) ;

%Shape: Rectangle [id:dp1132546813178934] 
\draw  [color={rgb, 255:red, 74; green, 144; blue, 226 }  ,draw opacity=1 ][fill={rgb, 255:red, 126; green, 211; blue, 33 }  ,fill opacity=0.44 ] (478,231) -- (508,231) -- (508,261) -- (478,261) -- cycle ;
%Shape: Rectangle [id:dp09147105437310299] 
\draw  [color={rgb, 255:red, 74; green, 144; blue, 226 }  ,draw opacity=1 ][fill={rgb, 255:red, 126; green, 211; blue, 33 }  ,fill opacity=0.44 ] (478,201) -- (508,201) -- (508,231) -- (478,231) -- cycle ;
%Shape: Rectangle [id:dp12135542008254308] 
\draw  [color={rgb, 255:red, 74; green, 144; blue, 226 }  ,draw opacity=1 ][fill={rgb, 255:red, 126; green, 211; blue, 33 }  ,fill opacity=0.44 ] (508,201) -- (538,201) -- (538,231) -- (508,231) -- cycle ;
%Shape: Rectangle [id:dp8650727329196213] 
\draw  [color={rgb, 255:red, 74; green, 144; blue, 226 }  ,draw opacity=1 ][fill={rgb, 255:red, 126; green, 211; blue, 33 }  ,fill opacity=0.44 ] (538,201) -- (568,201) -- (568,231) -- (538,231) -- cycle ;
%Shape: Rectangle [id:dp41225229921000994] 
\draw  [color={rgb, 255:red, 74; green, 144; blue, 226 }  ,draw opacity=1 ][fill={rgb, 255:red, 126; green, 211; blue, 33 }  ,fill opacity=0.44 ] (538,171) -- (568,171) -- (568,201) -- (538,201) -- cycle ;

% Text Node
\draw (537,103) node   {$\langle 7,\color{rgb, 255:red, 126; green, 211; blue, 33 }2\color{black},\color{rgb, 255:red, 126; green, 211; blue, 33 }4\color{black},1,9,\color{rgb, 255:red, 126; green, 211; blue, 33 }6\color{black},3,5,\color{rgb, 255:red, 126; green, 211; blue, 33 }8\color{black}\rangle$};
% Text Node
\draw (48,320) node   {$0$};
% Text Node
\draw (78,320) node   {$1$};
% Text Node
\draw (108,320) node   {$2$};
% Text Node
\draw (138,320) node   {$3$};
% Text Node
\draw (168,320) node   {$4$};
% Text Node
\draw (198,320) node   {$5$};
% Text Node
\draw (228,320) node   {$6$};
% Text Node
\draw (258,320) node   {$7$};
% Text Node
\draw (288,320) node   {$8$};
% Text Node
\draw (318,320) node   {$9$};
% Text Node
\draw (48,320) node   {$0$};
% Text Node
\draw (38,304) node   {$0$};
% Text Node
\draw (38,272) node   {$1$};
% Text Node
\draw (38,244) node   {$2$};
% Text Node
\draw (38,212) node   {$3$};
% Text Node
\draw (38,184) node   {$4$};
% Text Node
\draw (38,152) node   {$5$};
% Text Node
\draw (38,124) node   {$6$};
% Text Node
\draw (38,92) node   {$7$};
% Text Node
\draw (38,62) node   {$8$};
% Text Node
\draw (38,33) node   {$9$};
% Text Node
\draw (495,244) node   {$1$};
% Text Node
\draw (525,244) node   {$0$};
% Text Node
\draw (555,244) node   {$0$};
% Text Node
\draw (555,214) node   {$0$};
% Text Node
\draw (524,185) node   {$0$};
% Text Node
\draw (494,185) node   {$0$};
% Text Node
\draw (495,214) node   {$1$};
% Text Node
\draw (525,214) node   {$1$};
% Text Node
\draw (552,185) node   {$1$};

\end{tikzpicture}

\caption{An array $\langle 7, 2, 4, 1, 9, 6, 3, 5, 8\rangle$ is mapped to the 2D plane. An \textsf{LIS} is shown by green points. The plane is divided into a $3 \times 3$ grid. The number on each cell is the equal to the contribution of that cell to the \textsf{LIS}. The score of the grid is equal to the \textsf{LIS} of the array. The score of the grid is made by the path colored in green.} \label{fig:lis-grid}
\end{figure}